\documentclass[12pt]{amsart}
\usepackage[alphabetic]{amsrefs}
\usepackage{mathdots, blkarray}
\usepackage{multicol}
\usepackage{graphicx}
\usepackage{amscd}
\usepackage{upgreek}
\usepackage{stmaryrd}
\usepackage{longtable}
\usepackage[T1]{fontenc}
\usepackage{latexsym, amsmath, amssymb, amsthm}
\usepackage[svgnames]{xcolor}
\usepackage{rsfso}
\usepackage[mathscr]{eucal}
\usepackage{mathtools}
\usepackage{mathptmx}
\usepackage{titletoc}
\usepackage{wrapfig}
\usepackage{float}
\usepackage{xypic}
\usepackage{microtype}
\usepackage{dsfont}
\usepackage{xcolor}
\usepackage{color}
\usepackage[colorlinks = true,
            linkcolor  = DarkBlue,
            urlcolor   = DarkRed,
            citecolor  = DarkGreen]{hyperref}
\usepackage{hyperref}
\allowdisplaybreaks	
\usepackage[centering, includeheadfoot, hmargin=1.0in, tmargin=0.5in, 
  bmargin=0.6in, headheight=6pt]{geometry}
\newtheorem{theorem}{Theorem}[section]
\newtheorem{prop}[theorem]{Proposition}
\newtheorem{algm}[theorem]{Algorithm}

\newtheorem{lem}[theorem]{Lemma}
\newtheorem{coro}[theorem]{Corollary}

\newtheorem{thm}[theorem]{Theorem}

\newtheorem{rem}[theorem]{\rm\textsc{Remark}}
\newtheorem{exam}[theorem]{\rm\textsc{Example}}
\newcommand{\ideal}[1]{\ensuremath{\left\langle #1 \right\rangle}}
\newcommand{\cod}[1]{\ensuremath{\left| #1 \right\rangle}}

\newcommand{\oeq}[1]{\ensuremath{\overset{(\ref{#1})}{=}}}
\newcommand{\bslash}{\kern-0.1em\texttt{\scalebox{0.6}[1]{/}}\kern-0.15em \texttt{\scalebox{0.6}[1]{/}}}

\DeclareMathOperator{\GL}{GL}

\DeclareMathOperator{\dia}{diag}
\DeclareMathOperator{\wt}{wt}

\DeclareMathOperator{\swt}{swt}

\DeclareMathOperator{\Tr}{Tr}
\DeclareMathOperator{\tr}{tr}
\newcommand{\A}{\mathcal{A}} 
\newcommand{\B}{\mathcal{B}} 
\newcommand{\CC}{\mathcal{C}} 
\newcommand{\C}{\mathbb{C}} 
 
\newcommand{\Z}{\mathbb{Z}} 
\newcommand{\N}{\mathbb{N}} 
\newcommand{\F}{\mathbb{F}}

\newcommand{\ta}{\textbf{a}} 
\newcommand{\tb}{\textbf{b}} 
\newcommand{\tx}{\textbf{x}} 
 
\newcommand{\ra}{\longrightarrow}

\newcommand{\hbo}{$\hfill\Diamond$}

\begin{document}
\title{An invariant-theoretic approach to three weight enumerators\\ of self-dual quantum codes} 
\def\shorttitle{An invariant-theoretic approach to weight enumerators of self-dual quantum codes}

\author{Yin Chen}
\address{School of Mathematics and Physics, Key Laboratory of ECOFM of 
Jiangxi Education Institute, Jinggangshan University,
Ji'an 343009, Jiangxi, China \& Department of Finance and Management Science, University of Saskatchewan, Saskatoon, SK, Canada, S7N 5A7}
\email{yin.chen@usask.ca}

\author{Shan Ren}
\address{School of Mathematics and Statistics, Changchun University, Changchun 130022, China}
\email{rens734@nenu.edu.cn}

\author{Runxuan Zhang}
\address{Department of Mathematics and Information Technology, Concordia University of Edmonton, Edmonton, AB, Canada, T5B 4E4}
\email{runxuan.zhang@concordia.ab.ca}

\begin{abstract}
This article is a continuation of our recent work \cite{CZ24} in the setting of quantum error-correcting codes.
We use algebraic invariant theory to study three weight enumerators of formally self-dual quantum codes over arbitrary finite fields.  We derive a quantum analogue of Gleason's theorem, demonstrating that the weight enumerator of a formally self-dual quantum code can be expressed algebraically by two polynomials. We also show that the double weight enumerator of a formally self-dual quantum code can be expressed algebraically by five polynomials. We explicitly compute the complete weight enumerators of some special self-dual quantum codes. Our approach  illustrates the potential of employing algebraic  invariant theory to compute weight enumerators of self-dual quantum codes. 
\end{abstract}

\date{\today}
%\subjclass[2010]{13A50.}
\maketitle \baselineskip=17pt

%%%%%%%%%%%%%%%%%%%%%%%%%%%Contents%%%%%%%%%%%%%%%%%%%%%%%%
%\textcolor{blue}{\tableofcontents{}}
\dottedcontents{section}[1.16cm]{}{1.8em}{5pt}
\dottedcontents{subsection}[2.00cm]{}{2.7em}{5pt}
%\dottedcontents{subsubsection}[2.86cm]{}{3.4em}{5pt}

%%%%%%%%%%%%%%%%%%%%%%%%%%%Sections%%%%%%%%%%%%%%%%%%%%%%%%
\section{Introduction}
\setcounter{equation}{0}
\renewcommand{\theequation}
{1.\arabic{equation}}
\setcounter{theorem}{0}
\renewcommand{\thetheorem}
{1.\arabic{theorem}}

\noindent Quantum error-correcting codes and quantum error correction stem from \cite{SH95,Ste96a,Ste96b,KL97} and
have been studied extensively in the past 30 years because of the significant role  they play in analyzing physical principles, protecting information-carrying quantum states against decoherence, and making fault-tolerant quantum computation possible. It is well known that the weight distribution of a quantum code is fundamental in determining bounds on its minimum distance and hence its error-detecting capability, and that the computation of weight enumerators is indispensable for a thorough understanding of the structural and performance properties of quantum codes; see for example, \cite{AL99, AK01}, and \cite{KKKS06}. 
As a key component in computing weight enumerators, several MacWilliams-type identities for quantum codes have been obtained in \cite{SL97, Rai00, HYY19} and \cite{HYY20}. The primary objective of this article is to use algebraic invariant theory and MacWilliams identities to compute three kinds of weight enumerators of formally self-dual quantum codes.

As a classical topic in modern algebra, algebraic invariant theory begins with a faithful representation of a group and aims to 
study the subring of all polynomials fixed under the action of the group. Invariant theory of finite groups is an important tool for computing  weight (or shape) enumerators of self-dual codes in the classical coding theory; see \cite{Slo77, SA20} and \cite{NRS06} for a general reference of self-dual codes and invariant theory. The classical MacWilliams identity states 
that the weight enumerator $W_{\CC^\bot}(x,y)$ of the dual code $\CC^\bot$ of a code $\CC$ can be written as the image of the weight enumerator $W_\CC(x,y)$ of $\CC$ under the linear action of a finite group $G$, which means that
the weight enumerator $W_\CC(x,y)$ of a self-dual code $\CC$ can be viewed as a polynomial invariant under
the action of $G$. Together with MacWilliams identities, algebraic invariant theory has derived many substantial ramifications in computing 
weight enumerators of classical self-dual codes; see for example, \cite{Gle71, SA20} and \cite{CZ24}.

The notions of (Shor-Laflamme) weight distributions and weight enumerators for quantum codes were introduced by \cite{SL97}, which also derives a quantum MacWilliams identity and inspires numerous subsequent research work. For example, \cite{LHL16} derived a MacWilliams identity for quantum convolutional codes, and \cite{HESG18} used 
the quantum MacWilliams identity obtain some new bounds on the existence of absolutely maximally entangled quantum states in dimension larger than two.
Recently, the concept of Shor-Laflamme weight distributions was generalized to double weight distributions and complete weight distributions in \cite{HYY19}, and two MacWilliams identities about the double and complete weight enumerators have been developed for all finite fields $\F_q$; see \cite[Theorem 5]{HYY20}. In particular, they demonstrated how to use MacWilliams identities to determine the Singleton-type and Hamming-type bounds for arbitrary asymmetric quantum codes; see \cite[Theorems 1 and 2]{HYY20}. Since then, these three kinds of weight enumerators have attracted significant attention within the quantum coding theory community; see \cite{CS25} and \cite{KL25}. 

In this article, we take a viewpoint of algebraic invariant theory to explore the Shor-Laflamme weight enumerators, 
double weight enumerators, and complete weight enumerators for formally self-dual quantum codes. We say that a quantum code $Q$ is \textit{formally self-dual} if the two complete weight enumerators $D(M)$ and $D^\bot(M)$ are equal up to a nonzero scalar; see (\ref{FSD}) below. The Bell state code $Q_B$ is an example of formally self-dual codes; see Example \ref{exam2.3}. Note that some interesting examples of self-dual quantum codes respect to the Shor-Laflamme weight enumerators have appeared in the theory of quantum codes; see \cite[Example 7.2]{BCH23} and \cite[Example 3.1]{HYY20}.  We will see that the MacWilliams identities imply that the (Shor-Laflamme, double, or complete) weight enumerators of a formally self-dual quantum code can be regarded as invariant polynomials under certain group actions $(G,V)$.
This means that describing the general shape of these three weight enumerators is equivalent to 
computing the corresponding invariant rings $\C[V]^G$. Moreover, finding a homogeneous generating set for a given
invariant ring could be extremely challenging but it is the core task in algebraic invariant theory; see for example \cite{DK15} or \cite{CW11}. Many classical techniques from invariant theory, such as Molien's formula and Noether's bound theorem, play a significant role in computing the weight enumerators of classical self-dual codes; see \cite{Slo77} and \cite{SA20}. 

The present article can be viewed as a continuation of our recent work \cite{CZ24} to the theory of quantum codes. To our knowledge, this is the first application of algebraic invariant theory to characterize the structures of the  (Shor-Laflamme) weight enumerators, the double weight enumerators, and the complete weight enumerators of formally self-dual quantum codes. Roughly speaking, the first step of our approach is 
to determine the corresponding group $G$ and the representation $V$ via the MacWilliams identity; 
the second step is to find another representation $W$ of $G$ that is equivalent to $V$ but makes 
$\C[W]^G$ easier to compute; in the third step, we may use some suitable invertible matrix $T$
to transfer a homogeneous generating set $\A$ of $\C[W]^G$ to a homogeneous 
generating set $\B$ of $\C[V]^G$. In the case of classical self-dual linear codes, our method 
has successfully described the algebraic structures of the shape enumerators of self-dual NRT linear codes over any finite fields; see \cite[Section 4]{CZ24}.

We organize this article as follows.  Section \ref{sec2} contains fundamental facts and concepts about weight enumerators of quantum codes, including nice error basis and error groups, (double and complete) weight distributions, (double and complete)  weight enumerators, and three MacWilliams identities. 
In Section \ref{sec3}, we present a quick introduction to invariant theory of finite groups 
and show that the (Shor-Laflamme) weight enumerator $B(x,y)$ of a formally self-dual quantum code can be expressed 
by two algebraic independent invariant polynomials of $S_2$, the symmetric group of degree $2$; see
Corollary \ref{coro3.5}. This also drives a quantum  analogue of the famous Gleason's theorem (See \cite{Gle71} or \cite[Theorem 3c]{Slo77}).  Section \ref{sec4}  describes the double weight enumerators $C(x,y,z,w)$ of 
formally self-dual quantum codes. We show  in Corollary \ref{coro4.6} that the corresponding invariant ring $\C[V]^G$
is of Krull dimension $3$ and generated by five invariant polynomials $\{g_1,g_2,\dots,g_5\}$. As a direct consequence, the double weight enumerators $C(x,y,z,w)$ of a formally self-dual quantum code can be expressed by 
$\{g_1,g_2,\dots,g_5\}$; see  Corollary \ref{coro4.8}. In Section \ref{sec5},
we provide two explicit examples, computing the complete weight enumerator $D(M)$ of 
a formally self-dual quantum code for $q=2$ and $3$.  Our results show that 
the invariant ring $\C[V]^G$ for $q=2$ is a polynomial algebra as well as $\C[V]^G$ for $q=3$ is not polynomial but it is isomorphic to the tensor product of a polynomial algebra and the second Veronese subring of the polynomial ring in three variables; see Theorems \ref{thm5.1} and  \ref{thm5.2}, respectively. These two examples also clearly demonstrate the complexity involved in computing the complete weight enumerators of self-dual quantum codes.

Throughout this article, we assume that $n\in\N^+:=\{1,2,3,\dots\}$. We use $I_n$ to denote the identity matrix of degree $n$;  write $\GL_n(\C)$ for the general linear group of degree $n$ over the 
complex field $\C$; and denote by $S_n$ the symmetric group of degree $n$. 
We write $A^t$ for the transpose of a matrix (or a vector) $A$.

\section{Weight Enumerators of Quantum Codes}\label{sec2}
\setcounter{equation}{0}
\renewcommand{\theequation}
{2.\arabic{equation}}
\setcounter{theorem}{0}
\renewcommand{\thetheorem}
{2.\arabic{theorem}}

\noindent In this preliminary section, we recap some fundamental concepts and facts about quantum error-correcting codes, weight enumerators, and the MacWilliams identities. Let $p$ be a prime and $\F_q$ be a finite field of order $q=p^s$ for some $s\in\N^+$. Let $\upzeta_p:=e^{\frac{2\uppi \sqrt{-1}}{p}}$ be a primitive $p$-th root of unity and $\tr$
denote the trace map from $\F_q$ to $\F_p$, i.e.,
\begin{equation}
\label{ }
\tr(a):=\sum_{i=0}^{s-1} a^{p^i}
\end{equation}
for all $a\in\F_q$. We write $\Tr(A)$ for the trace of a linear operator $A$ and $B^\dag$ denotes the Hermitian transpose of a complex unitary linear operator $B$.

%see for example, \cite[Theorem 7.12]{Wan12} for fundamental properties of the trace map.

\subsection{Error groups}

Consider the state space $\C^q$ of a quantum mechanical system, which may be regarded as a $q$-dimensional Hilbert space over $\C$ equipped with the Hermitian inner product. Let us fix an orthonormal basis $\{\cod x\mid x\in\F_q\}$ of $\C^q$. The tensor product $(\C^q)^{\otimes n}$ of $n$ copies of $\C^q$ is used to transmit
$n$ qubits of information. Recall that a \textit{quantum} (\textit{error-correcting}) \textit{code} $Q$ is a nonzero $K$-dimensional subspace of $\C^{q^n}=(\C^q)^{\otimes n}$, where $K=\dim_\C(Q)\in\N^+$ and $n$ is called the \textit{length} of $Q$.  Vectors in $Q$ are called \textit{codewords}. Besides the length $n$ and the dimension $K$, an additional fundamental parameter of $Q$ is the minimum distance $d$, which measures the error-detecting capability of $Q$. Accordingly, $Q$ is called an \textit{$((n,K,d))_q$-code} if it has minimum distance $d$. In some situations, when the minimum distance $d$ is not under consideration, we also refer to $Q$ as an \textit{$((n,K))_q$-code}.

To evaluate the performance of a code $Q$, an appropriate error model must be specified. Here we choose the error model appeared in \cite[Section 2]{KKKS06}, which has been used extensively in the literature. 
We briefly recap the basic facts and constructions associated with this model. Given arbitrary two elements $a,b\in\F_q$, we may define two unitary operators $X_a$ and $Z_b$ on $\C^q$:
\begin{equation}
\label{ }
X_a:\cod x\mapsto\cod{x+a}\textrm{ and }Z_b: \cod x \mapsto \upzeta_p^{\tr(bx)}\cod x.
\end{equation}
The following set formed by $X_a$ and $Z_b$
\begin{equation}
\label{ }
E:=\{X_aZ_b\mid a,b\in\F_q\}
\end{equation}
is called a \textit{set of error operators}. By \cite[Lemma 1]{KKKS06}, we see that the set $E$ is a nice error basis on 
$\C^q$. In other words, (1) $E$ contains the identity map; (2) the composition of any two elements in $E$ is closed up to
a scalar; (3) $\Tr(A^\dag B)=0$ for all distinct $A,B\in E$.

Note that the tensor product of two nice error bases is also a nice error basis; see  \cite[Lemma 3]{KKKS06}.
To construct a nice error basis on $\C^{q^n}=(\C^q)^{\otimes n}$, we write 
\begin{equation}
\label{ }
\left\{\cod{\tx}=\cod{x_1}\otimes \cod{x_2}\otimes \cdots\otimes \cod{x_n}\mid \tx=(x_1,x_2,\dots,x_n)\in\F_q^n\right\}
\end{equation}
for a basis of $\C^{q^n}$. Let $\textbf{a}=(a_1,\dots,a_n)$ and $\textbf{b}=(b_1,\dots,b_n)$ be two vectors in $\F_q^n$. 
We may define
\begin{equation}
\label{ }
X_\ta:=X_{a_1}\otimes \cdots \otimes X_{a_n}\textrm{ and }
Z_\tb:=Z_{b_1}\otimes \cdots \otimes Z_{b_n}.
\end{equation}
It follows from \cite[Corollary 4]{KKKS06} that the set 
\begin{equation}
\label{ }
E_n:=\{X_\ta Z_\tb=X_{a_1}Z_{b_1}\otimes \cdots \otimes X_{a_n}Z_{b_n}\mid \ta, \tb\in\F_q^n\}
\end{equation}
is a nice error basis on $\C^{q^n}$, and also can be viewed as a basis for the vector space of all 
$q^n\times q^n$-matrices over $\C$. Moreover, one observes that
\begin{eqnarray}
X_\ta  \cod{\tx}& = & \cod{\tx+\ta} \\
Z_\tb  \cod{\tx} & = &  \upzeta_p^{\tr(\tx\cdot \tb)}\cod{\tx}
\end{eqnarray}
where $\tx\cdot \tb:=\sum_{i=1}^nx_ib_i\in\F_q$; see \cite[Section 2]{HYY20}.

The group $G_n$, generated by all elements in $E_n$, is called the \textit{error group}
associated with the nice error basis $E_n$. For any $a,a',b,b'\in\F_q$, the following identity
\begin{equation}
\label{ }
X_aZ_bX_{a'}Z_{b'}=\upzeta_p^{\tr(ba')}X_{a+a'}Z_{b+b'}
\end{equation}
which has been verified in the proof of \cite[Lemma 1]{KKKS06}, shows that
\begin{equation}
\label{ }
G_n=\{\upzeta_p^c X_\ta Z_\tb\mid \ta,\tb\in\F_q^n,c\in\F_p\}.
\end{equation}
Clearly, $G_n$ is a finite group of order $pq^{2n}$, not necessarily abelian. A \textit{stabilizer code}  is a nonzero subspace $Q$
of $\C^{q^n}$ such that
\begin{equation}
\label{ }
Q=\bigcap_{e\in S}\{v\in \C^{q^n}\mid e(v)=v\}
\end{equation}
for some subgroup $S$ of $G_n$.

\begin{exam}[Bell state code]\label{Bell1}{\rm
Consider $q=p=2$ and $n=2$. We may write $\{\cod{00},\cod{01},
\cod{10},\cod{11}\}$ for a basis for $\C^4=(\C^2)^{\otimes 2}$. Then 
$$X_0=Z_0=I_2,~ X_1=\begin{pmatrix}
     0 & 1   \\
     1 &  0
\end{pmatrix},~ Z_1=\begin{pmatrix}
     1 & 0   \\
     0 &  -1
\end{pmatrix},\textrm{ and } X_1Z_1=\begin{pmatrix}
     0 & -1   \\
     1 &  0
\end{pmatrix}$$
form the set $E$ of error operators.  Hence,
$$
E_2=\{X_{i}Z_{j} \otimes  X_{s}Z_{t}\mid i,j,s,t\in\{0,1\}\}
$$
is a subset of $4\times 4$-matrices, and has order $2^4=16$. Elements in $E_2$ can be expressed 
as the Kronecker product of two $2\times 2$-matrices.  For instance, 
$$X_{0}Z_{1} \otimes  X_{1}Z_{0}=Z_{1} \otimes  X_{1}=\begin{pmatrix}
     1 & 0   \\
     0 &  -1
\end{pmatrix}\otimes \begin{pmatrix}
     0 & 1   \\
     1 &  0
\end{pmatrix}=\begin{pmatrix}
     0 &1&0&0    \\
     1 &0&0&0\\  
    0 &0&0&-1\\
    0 &0&-1&0
\end{pmatrix}.$$
Moreover, the 1-dimensional quantum code $Q_B$, spanned by 
$$\frac{1}{\sqrt{2}}(\cod{00}+\cod{11})$$
is called the \textit{Bell state code}. 
\hbo}\end{exam}

\subsection{Weight enumerators of quantum codes}

Recall that the \textit{symplectic weight} of a vector $(\ta|\tb)$ in $\F_q^{2n}$ is defined as:
\begin{equation}
\label{ }
\swt((\ta|\tb)):=\# \{1\leqslant i\leqslant n\mid (a_i,b_i)\neq (0,0)\},
\end{equation}
and one defines  the \textit{weight} $\wt(e)$ of an element $e=\upzeta^cX_\ta Z_\tb \in G_n$ as
\begin{equation}
\label{ }
\wt(e):=\swt((\ta|\tb)).
\end{equation}
In particular, the weight of a scalar multiple of the identity map is zero. We fix an ordering 
$\upalpha_0=0,\upalpha_1,\dots,\upalpha_{q-1}$ for the elements in $\F_q$. For $\uplambda,\upmu\in\{0,1,\dots,q-1\}$, we use $N_{\uplambda,\upmu}(e)$ to denote 
the number of the operator $X_{\upalpha_\uplambda}Z_{\upalpha_\upmu}$ occurred in the expression of $e\in G_n$. Note that $\wt(e)$ reveals the number of non-identity tensor components of $e$. Thus $\wt(e)$ can be written as the sum of all $N_{\uplambda,\upmu}(e)$ for all $(\uplambda,\upmu)\neq (0,0)$. In other words,
\begin{equation}
\label{ }
\wt(e)=\sum_{(\uplambda,\upmu)\neq (0,0)}N_{\uplambda,\upmu}(e).
\end{equation}

Two related notions of weights are introduced in \cite{HYY20}. The \textit{$X$-weight} $\wt_X(e)$ of $e$ is defined as the sum of all 
$N_{\uplambda,\upmu}(e)$ with $\uplambda\neq 0$, as well as the \textit{$Z$-weight} $\wt_Z(e)$ of $e$ is defined as the sum of all $N_{\uplambda,\upmu}(e)$ with $\upmu\neq 0$. Thus
\begin{eqnarray}
\wt_X(e)&=&\wt(e)-\sum_{\upmu\neq 0} N_{0,\upmu}(e) \\
\wt_Z(e)&=&\wt(e)-\sum_{\uplambda\neq 0} N_{\uplambda, 0}(e).
\end{eqnarray}
Furthermore, for $i,j\in\N$, we define the following two error sets: 
\begin{equation}
\label{ }
E[i]:=\{e\in E_n\mid \wt(e)=i\} \textrm{ and } E[i,j]:=\{e\in E_n\mid \wt_X(e)=i, \wt_Z(e)=j\}.
\end{equation}
We write $\updelta(n)$ for the set consisting of all $q\times q$-matrices whose 
entries are non-negative integers with the total sum $n$. The error set
$E[J]$ associated to an index matrix $J=(J_{\uplambda,\upmu})\in\updelta(n)$ is defined as
\begin{equation}
\label{ }
E[J]:=\{e\in E_n\mid N_{\uplambda,\upmu}(e)=J_{\uplambda,\upmu},\textrm{ for all }\uplambda,\upmu=0,1,\dots,q-1\}.
\end{equation}

Now we consider an $((n,K))_q$-code $Q$ and use $P$ to denote the orthogonal projection from $\C^{q^n}$ to $Q$. 
The \textit{weight distributions} for $Q$, originally due to Shor and Laflamme,  are defined by the two sequences of numbers:
\begin{equation}
\label{ }
B_i:=\frac{1}{pK^2}\sum_{e\in E[i]}\Tr(e^\dag P)\Tr(eP)\textrm{ and } 
B_i^\bot:=\frac{1}{pK}\sum_{e\in E[i]}\Tr(e^\dag PeP),
\end{equation}
which correspond to two \textrm{weight enumerators} of $Q$:
\begin{equation}
\label{ }
B(x,y):=\sum_{i=0}^n B_i\cdot x^{n-i}y^i\textrm{ and }B^\bot(x,y):=\sum_{i=0}^n B_i^\bot\cdot x^{n-i}y^i.
\end{equation}
See for example, \cite{SL97} and \cite[Section 5]{KKKS06} for more details.
Replacing the the error set $E[i]$ by $E_{i,j}$ and $E[J]$, two new concepts of weight distributions of $Q$: 
double weight distribution and complete weight distribution, are introduced in \cite{HYY19} and \cite{HYY20}.
More precisely,  the \textit{double weight distributions} for $Q$ are defined by
\begin{equation}
\label{ }
C_{i,j}:=\frac{1}{pK^2}\sum_{e\in E[i,j]}\Tr(e^\dag P)\Tr(eP)\textrm{ and } 
C_{i,j}^\bot:=\frac{1}{pK}\sum_{e\in E[i,j]}\Tr(e^\dag PeP),
\end{equation}
and the corresponding \textit{double weight enumerators} are defined as
\begin{equation}
\label{ }
C(x,y,z,w):=\sum_{i,j=0}^n C_{ij}\cdot x^{n-i}y^iz^{n-j}w^j\textrm{ and } C^\bot(x,y,z,w):=
\sum_{i,j=0}^n C_{ij}^\bot\cdot x^{n-i}y^iz^{n-j}w^j.
\end{equation}
The  \textit{complete weight distributions} for $Q$ are defined by
\begin{equation}
\label{ }
D_{J}:=\frac{1}{pK^2}\sum_{e\in E[J]}\Tr(e^\dag P)\Tr(eP)\textrm{ and } 
D_{J}^\bot:=\frac{1}{pK}\sum_{e\in E[J]}\Tr(e^\dag PeP).
\end{equation}
The \textit{complete weight enumerators} of $Q$ can be expressed as polynomials associated with
a $q\times q$-matrix $M=(M_{\uplambda,\upmu})$:
\begin{equation}
\label{ }
D(M):=\sum_{J=(J_{\uplambda,\upmu})\in \updelta(n)} D_{J}\cdot M^J\textrm{ and } 
D^\bot(M):=\sum_{J=(J_{\uplambda,\upmu})\in \updelta(n)} D_{J}^\bot\cdot M^J,
\end{equation}
where $M^J$ is defined by $\prod_{\uplambda,\upmu\in\{0,1,\dots,q-1\}}M_{\uplambda,\upmu}^{J_{\uplambda,\upmu}}.$

We provide the following example to illustrate how to compute the complete weight (distributions) enumerators of quantum codes.

\begin{exam}\label{Bell2}{\rm
Let us continue to work on the Bell state code $Q_B$ appeared in Example \ref{Bell1}. With respect to the basis 
$\{\cod{00},\cod{01},\cod{10},\cod{11}\}$, the orthogonal projection $P:\C^4\ra\C^4$ can be expressed by
$$P=\frac{1}{2}\begin{pmatrix}
    1  & 0&0 & 1  \\
    0  & 0 &0 &0\\
   0  & 0 &0 &0\\
     1 &0&0&1
\end{pmatrix}.$$
The set $\updelta(2)$ of index matrices consists of the following $10$ matrices (in this ordering):
\begin{eqnarray*}
\begin{pmatrix}
    0  & 0   \\
     0 & 2 
\end{pmatrix}, \begin{pmatrix}
    0  & 0   \\
     1 & 1 
\end{pmatrix}, \begin{pmatrix}
    0  & 0   \\
     2& 0 
\end{pmatrix}, \begin{pmatrix}
    0  & 1  \\
     0 & 1
\end{pmatrix}, \begin{pmatrix}
    0  & 1   \\
     1 & 0 
\end{pmatrix}, \\
\begin{pmatrix}
    0  & 2   \\
     0 & 0 
\end{pmatrix}, \begin{pmatrix}
    1  & 0   \\
     0 & 1 
\end{pmatrix}, \begin{pmatrix}
    1  & 0   \\
     1& 0 
\end{pmatrix}, \begin{pmatrix}
    1  & 1  \\
     0 & 0
\end{pmatrix}, \begin{pmatrix}
    2 & 0   \\
     0 & 0 
\end{pmatrix}.
\end{eqnarray*}
For all matrices $J\in\updelta(2)$ in the ordering above, i.e., $J_1=\begin{pmatrix}
    0  & 0   \\
     0 & 2 
\end{pmatrix},\dots,J_{10}=\begin{pmatrix}
    2 & 0   \\
     0 & 0 
\end{pmatrix}$, the cardinalities of the corresponding error sets $E(J)$ are
$$\{1,2,1,2,2,1,2,2,2,1\}$$
respectively.  For example, to determine the elements of $E(J_1)$, we assume that $e\in E(J_1)$ denotes 
an any element. Note that
$N_{0,0}(e)=N_{0,1}(e)=N_{1,0}(e)=0$ and $N_{1,1}(e)=2$. Thus $X_1Z_1$ appears twice in the expression of $e$. 
This implies that 
$$e=X_1Z_1\otimes X_1Z_1=\begin{pmatrix}
     0 & -1   \\
     1 &  0
\end{pmatrix}\otimes \begin{pmatrix}
     0 & -1   \\
     1 &  0
\end{pmatrix}=\begin{pmatrix}
     0 &0&0&1    \\
     0 &0&-1&0\\  
    0 &-1&0&0\\
    1 &0&0&0
\end{pmatrix}.$$
A similar method can be applied to determine the elements in $E(J_i)$ for all $i=1,2,\dots,10$. Moreover, 
a direct computation shows that
$D_{J_i}=\frac{1}{2}=D^\bot_{J_i}$ for $i\in\{1,3,6,10\}$ and $D_{J_i}=0=D^\bot_{J_i}$ for others. 
Hence, we have
\begin{eqnarray*}
D^\bot(M)&=&D(M)=D_{J_1}\cdot M^{J_1}+D_{J_3}\cdot M^{J_3}+D_{J_6}\cdot M^{J_6}+D_{J_{10}}\cdot M^{J_{10}}\\
&=&\frac{1}{2}\left(M_{11}^2 + M_{10}^2 + M_{01}^2 + M_{00}^2\right)
\end{eqnarray*}
where $M=\begin{pmatrix}
   M_{00}   & M_{01}   \\
    M_{10}  & M_{11} 
\end{pmatrix}$ denotes the matrix of variables.
\hbo}\end{exam}

More examples that illustrate the notions above and compute double weight enumerators of quantum codes can be found in \cite[Section 3]{HYY19} and \cite[Example 1]{HYY20}.

\subsection{MacWilliams identities}

Let $Q$ be an $((n,K))_q$-code. Several MacWilliams identities remain to hold for the
weight enumerators $B,B^\bot$, the double weight enumerators $C,C^\bot$, and the complete weight enumerators $D,D^\bot$ for $Q$:
\begin{eqnarray}
B(x,y)&=&\frac{1}{q^n\cdot K}\cdot B^{\bot}\left(x+(q^2-1)y,x-y\right); \label{Mac1}\\
C(x,y,z,w)&=&\frac{1}{K}\cdot C^{\bot}\left(x+(q-1)y,x-y,\frac{z+(q-1)w}{q},\frac{z-w}{q}\right); \label{Mac2}\\
D(M)&=& \frac{1}{K}\cdot D^{\bot}(M^\bot), \label{Mac3}
\end{eqnarray}
where $M=(M_{\uplambda,\upmu})$ and $M^\bot=(M_{\uplambda',\upmu'}^\bot)$ denote
$q\times q$-matrices with entries 
$$M_{\uplambda',\upmu'}^\bot=\frac{1}{q}\cdot\sum_{\uplambda,\upmu\in\{0,1,\dots,q-1\}}
\upzeta_p^{\tr(\upalpha_{\uplambda'}\cdot \upalpha_{\upmu}-\upalpha_{\uplambda}\cdot \upalpha_{\upmu'})} M_{\uplambda,\upmu}.$$
See \cite[Theorem 5]{HYY20} for the detailed proofs of these three identities.

We say that an $((n,K))_q$-code $Q$ is  \textit{formally self-dual} if its complete weight enumerators satisfy the following  relation: 
\begin{equation}
\label{FSD}
D(M)=\frac{1}{K}\cdot D^\bot(M).
\end{equation}

\begin{exam}\label{exam2.3}{\rm
By Examples  \ref{Bell1} and \ref{Bell2}, the Bell state code $Q_B$ is a formally self-dual code of the type $((2,1))_2$, with the complete weight enumerator $D(M)=\frac{1}{2}\left(M_{11}^2 + M_{10}^2 + M_{01}^2 + M_{00}^2\right)$.
\hbo}\end{exam}

\begin{prop}\label{prop2.4}
If  an $((n,K))_q$-code $Q$ is formally self-dual, then
\begin{equation}\label{eq2.16}
B(x,y)=\frac{1}{K}\cdot B^\bot(x,y)\textrm{ and }C(x,y,z,w)=\frac{1}{K}\cdot C^\bot(x,y,z,w).
\end{equation}
\end{prop}

\begin{proof}
We write $\upPhi(x,y)$ for the $q\times q$-matrix for which the entry at the first column and
first row is $x$, and other entries all are $y$. By \cite[Theorem 4, (3.1)]{HYY20}, we see that
$K\cdot B(x,y)=K\cdot D(\upPhi(x,y))\oeq{FSD} D^\bot(\upPhi(x,y))=B^\bot(x,y)$, where the last equation holds from
\cite[Theorem 4, (3.2)]{HYY20}. This shows that $B(x,y)=\frac{1}{K}\cdot B^\bot(x,y).$ Moreover, write $\upPsi(x,y,z,w)$
for the $q\times q$-matrix in which the first row is $(xz,yz,\dots,yz)$, the first column is $(xz,xw,\dots,xw)^t$, and other entries are $yw$. The equation (3.3) in \cite[Theorem 4]{HYY20} implies that
$K\cdot C(x,y,z,w)=K\cdot D(\upPsi(x,y,z,w))\oeq{FSD} D^\bot(\upPsi(x,y,z,w))=C^\bot(\upPsi(x,y,z,w))$, where the last equation follows from the equation (3.4) in \cite[Theorem 4]{HYY20}. Thus, $C(x,y,z,w)=\frac{1}{K}\cdot C^\bot(x,y,z,w).$
\end{proof}

We close this section with the following remark that illustrates how to use the language of group actions and invariant theory to understand the weight enumerator $B(x,y)$ of a formally self-dual code $Q$. We will take the same language to examine the algebraic properties of the double  weight enumerator $C(x,y,z,w)$ and the complete weight enumerator $D(M)$ for $Q$ in Sections \ref{sec4} and \ref{sec5}.

\begin{rem}{\rm
Suppose that $Q$ denotes a formally self-dual code of type $((n,K))_q$. Let us consider the weight enumerator $B(x,y)$ of $Q$.
Note that $B(x,y)$ is a homogeneous polynomial of degree $n$ in $\C[x,y]$.
By the MacWilliams identity (\ref{Mac1}) and (\ref{eq2.16}), we see that
\begin{eqnarray*}
B(x,y)&=&\frac{1}{q^n\cdot K}\cdot B^{\bot}(x+(q^2-1)y,x-y)\\
&=&\frac{1}{q^n\cdot K}\cdot K\cdot B(x+(q^2-1)y,x-y)\\
&=&B\left(\frac{1}{q}\cdot x+\frac{q^2-1}{q}\cdot y,\frac{1}{q}\cdot x-\frac{1}{q}\cdot y\right).
\end{eqnarray*}
We  define $$\upsigma:=\begin{pmatrix}
   \frac{1}{q}   &  \frac{q^2-1}{q} \\
   \frac{1}{q}    &  -\frac{1}{q}
\end{pmatrix}\in\GL_2(\C).$$ A direct verification shows that
$\upsigma^2=I_2$ and so $\upsigma^{-1}=\upsigma$. Using the language of group actions and invariant theory (see Section \ref{sec3} below for details), we have
\begin{eqnarray*}
\upsigma\cdot B(x,y)&=&B\left((\upsigma^{-1})^t(x),(\upsigma^{-1})^t(y)\right)=B(\upsigma^t(x),\upsigma^t(y))\\
&=&B\left(\frac{1}{q}\cdot x+\frac{q^2-1}{q}\cdot y,\frac{1}{q}\cdot x-\frac{1}{q}\cdot y\right)\\
&=&B(x,y)
\end{eqnarray*}
where $\upsigma^t$ denotes the transpose of $\upsigma$. This means that the enumerator $B(x,y)$ can be viewed as a polynomial invariant under the action of the cyclic group of order $2$ generated by $\upsigma$. Therefore, the question of finding all possible $B(x,y)$ for formally self-dual codes is equivalently to a question of computing polynomial invariants of a finite group. 
\hbo}\end{rem}

\section{Weight Enumerators of Self-dual Quantum Codes}\label{sec3}
\setcounter{equation}{0}
\renewcommand{\theequation}
{3.\arabic{equation}}
\setcounter{theorem}{0}
\renewcommand{\thetheorem}
{3.\arabic{theorem}}

\noindent This section aims to use the language of representation theory and invariant theory  to describe 
the  weight enumerator $B(x,y)$ of a formally self-dual quantum code $Q$ of the type $((n,K))_q$, and to derive a quantum analogue of Gleason's theorem in Corollary \ref{coro3.5} below. 

\subsection{Polynomial invariant theory}

Let $G$ be a group (not necessarily a finite group) and $V$ be a finite-dimensional faithful representation of $G$ over a field $k$.
Let $V^*$ denote the dual space of $V$ and $k[V]$ denote the symmetric algebra on $V^*$. The action of $G$ on $V^*$
can be extended algebraically to a $k$-linear action of $G$ on $k[V]$. After choosing  a basis $\{e_1,\dots,e_m\}$ for $V$
and a basis $\{x_1,\dots,x_m\}$ for $V^*$ dual to $\{e_1,\dots,e_m\}$, we may identify $k[V]$ with
the polynomial ring $k[x_1,\dots,x_m]$. The action of $G$ on $k[V]$ can be defined as
\begin{equation}
\label{ }
(\upsigma\cdot f)(v):=f(\upsigma^{-1}\cdot v)
\end{equation}
for all $\upsigma\in G, f\in k[V]$, and $v\in V$. This action is also degree-preserving, i.e., if $f\in k[V]$ is a 
homogeneous polynomial of degree $d$, then $\upsigma\cdot f$ is also homogeneous and has the same degree $d$.

More precisely, if $\upsigma=(a_{ij})_{m\times m}$ is a group element, we assume throughout this article that the action of $\upsigma$ on a polynomial $f(x_1,\dots,x_m)$ is given by
\begin{equation}
\begin{aligned}
\upsigma(x_1) & =  a_{11}x_1+a_{21}x_2+\cdots+a_{m1}x_m \\
\upsigma(x_2) & =  a_{12}x_1+a_{22}x_2+\cdots+a_{m2}x_m \\ 
\vdots &\hspace{1.5cm}\vdots \hspace{2cm}\vdots\\
\upsigma(x_m) & =  a_{1m}x_1+a_{2n}x_2+\cdots+a_{mm}x_m.
\end{aligned}
\end{equation}

\begin{exam}{\rm
Consider the matrix $\upsigma=\begin{pmatrix}
    1  & 0   \\
    1  &  1
\end{pmatrix}\in\GL_2(\C)$ and let $G$ be the subgroup of $\GL_2(\C)$ generated by $\upsigma$. Then
$\upsigma^i=\begin{pmatrix}
    1  & 0   \\
    i  &  1
\end{pmatrix}$ for all $i\in\Z$ and $G\cong (\Z,+)$ is an infinite cyclic group. We use $V$ with the basis  $\{e_1,e_2\}$ to denote the standard 
two-dimensional representation of $G$ over $\C$. Then
$\upsigma(e_1)=e_1+e_2$ and $\upsigma(e_2)=e_2$. Note that the resulting matrix of $\upsigma$ on $V^*$ is 
$$(\upsigma^{-1})^t=\begin{pmatrix}
    1  & -1   \\
    0  &  1
\end{pmatrix}.$$ Hence, $\upsigma(x_1)=x_1$ and $\upsigma(x_2)=-x_1+x_2$.
\hbo}\end{exam}

The subring $k[V]^G$ of $k[V]$ consisting of all polynomials fixed by the action of $G$ is called the \textit{invariant ring} of $G$ on $V$. Namely,
\begin{equation}
\label{ }
k[V]^G:=\{f\in k[x_1,\dots,x_m]\mid \upsigma(f)=f,\textrm{ for all }\upsigma\in G\}
\end{equation}
which is the main object of study in polynomial invariant theory. If $G$ is a finite group, a theorem due to Emmy Noether in 1923 states that $k[V]^G$ is a finitely generated $\N$-graded commutative algebra over $k$; see \cite[Proposition 3.0.1]{DK15} for a modern proof of this theorem. If $G$ is linearly reductive group and $V$ is a rational representation, the same conclusion remains to hold by the famous theorem due to David Hilbert;  see for example, \cite[Theorem 2.2.10]{DK15}.

To understand the algebraic structure of an invariant ring $k[V]^G$,  the study of two fundamental questions plays a core role 
 in invariant theory. The first question is about how to find a minimal generating set for $k[V]^G$, and the second one asks how to find a set of generating relations among these generators. Suppose that $G$ is finite. The invariant ring $k[V]^G$ is said to be \textit{modular} if the characteristic of $k$ divides the order of $G$; otherwise, \textit{nonmodular}.  The nonmodular case includes two subcases: (1)  the characteristic of $k$ is zero; (2) the characteristic of $k$ is positive but doesn't divide the order of $G$. Nonmodular invariant theory of finite groups has been understood very well while modular invariant theory is a challenging topic; see for example, \cite{DK15} or \cite{CW11} for general references of invariant theory of finite groups. 
 
 We close this subsection by recalling the following algorithm appeared in \cite{CZ24}, which illustrates how to transform the generating set of $k[W_1]^G$ to a generating set of $k[W_2]^G$ for two equivalent representations  $W_1$ and $W_2$ of a finite group $G$ over a field $k$. More precisely, suppose that
$\A$ denotes a homogeneous generating set of $k[W_1]^G$. Assume that $d$ denotes the maximal degree of elements of $\A$ and arrange the elements of $\A$ in the ascending order of degree, i.e., $\A=\{f_1,f_2,\dots,f_s\}$ with $\deg(f_1)\leqslant \deg(f_2)\leqslant \cdots\leqslant \deg(f_s)=d$.

\begin{algm} \label{algm}
We may construct a homogeneous generating set $\B$ for $k[W_2]^G$ from $\A$ by performing the following steps:
\begin{enumerate}
  \item Find $T\in GL_n(k)$ such that $g_{W_2}=T^{-1}\circ g_{W_1}\circ T$
for all $g\in G$;
  \item Let $\B:=\emptyset, f:=f_1$ and repeat Steps $(2) - (3)$ where $f$ runs over $\A$;
  \item Compute the invertible matrix $T_{\deg(f)}^{-1}$ and add the image $T_{\deg(f)}^{-1}\cdot f$ into $\B$;
  \item After $s$ steps, this algorithm terminates and $\B$ is a homogeneous generating set of $k[W_2]^G$.
\end{enumerate}
\end{algm}

\noindent See \cite[Algorithm 3.5]{CZ24} for a detailed proof of this algorithm.

\subsection{Weight enumerators}

We are able to use the invariant theory of finite groups to describe the weight enumerator $B(x,y)$ of an arbitrary formally self-dual code $Q$ of the type $((n,K))_q$. 

Throughout this subsection, we consider 
$$\upsigma=\begin{pmatrix}
   \frac{1}{q}   &    \frac{q^2-1}{q} \\
 \frac{1}{q}    &  -\frac{1}{q}
\end{pmatrix}
$$
and the cyclic subgroup $G=\ideal{\upsigma}$ of $\GL_2(\C)$, generated by $\upsigma$. As $\upsigma^2=I_2$, it follows that
$|G|=2$ and $G\cong S_2$, the symmetric group of degree $2$. 

We write $V$ for the $2$-dimensional standard representation of $G$ and $\{x,y\}$ for the dual basis of $V^*$. Thus
$\C[V]=\C[x,y]$ and $\C[V]^G=\C[x,y]^G$. Note that $\upsigma^{-1}=\upsigma$ and
\begin{equation}
\label{ }
\upsigma: x\mapsto \frac{1}{q} \cdot x +   \frac{q^2-1}{q}\cdot y\textrm{ and } \upsigma:y\mapsto\frac{1}{q} \cdot x -   \frac{1}{q}\cdot y.
\end{equation}
Define
\begin{equation}
\begin{aligned}
f_1 & :=  x+ (q-1)\cdot y\\
f_2 & :=  \left(x- (q+1) \cdot y\right)^2.
\end{aligned}
\end{equation}
Clearly, $f_1,f_2\in \C[V]^G$ both are $G$-invariant. Moreover, 

\begin{thm}\label{thm3.3}
$\C[V]^G=\C[f_1,f_2]$ is a polynomial algebra over $\C$.
\end{thm}

\begin{rem}{\rm
The standard method  to prove this statement in invariant theory is to apply the criterion appeared in 
\cite[Proposition 16]{Kem96}. 
First of all, we note that $V$ is a faithful representation and $|G|=|S_2|=2=\deg(f_1)\cdot \deg(f_2)$. Secondly, a direct computation verifies that the determinant of the Jacobian matrix of $\{f_1,f_2\}$ is nonzero, thus $f_1,f_2$ are algebraically independent over $\C$. Applying \cite[Proposition 16]{Kem96} shows that $\C[V]^G$ is a polynomial algebra over $\C$, generated by $\{f_1,f_2\}$.
\hbo}\end{rem}

However, below we would like to provide a constructive approach using Algorithm \ref{algm} to prove Theorem \ref{thm3.3}. 
This approach shows how we derive the generators $f_1$ and $f_2$ explicitly.

\begin{proof}[Proof of Theorem \ref{thm3.3}]
Let us define
$$T:=\begin{pmatrix}
   \frac{q+1}{q}   &  \frac{q^2-1}{q}  \\
  \frac{1-q}{q}    &  \frac{q^2-1}{q}
\end{pmatrix}.$$
A direct verification shows that
$$\upsigma=T^{-1}\cdot\begin{pmatrix}
   1   & 0   \\
     0 & -1 
\end{pmatrix}\cdot T.$$ Suppose that $H$ denotes the subgroup of $\GL_2(\C)$ generated by $\begin{pmatrix}
   1   & 0   \\
     0 & -1 
\end{pmatrix}$ and $W$ denotes the standard representation of $H$. Then $W\cong V$ are isomorphic  as 
$S_2$-representations. It is not difficult  to verify that $$\C[W]^H=\C[x,y]^H=\C[x,y^2].$$  By \cite[Proposition 3.1]{CZ24}, we see that
$\C[V]^G\cong \C[V]^H$ is a polynomial algebra over $\C$. Moreover, by  Algorithm \ref{algm},
the first $H$-invariant $x$, together with the action of $T$, produces a $G$-invariant:
$$\frac{q+1}{q}\cdot x+\frac{q^2-1}{q} \cdot y$$
which gives rise to the first generator $f_1=x+ (q-1)\cdot y$ of  $\C[V]^G$ via dividing the nonzero scalar $\frac{q+1}{q}$.
The second $H$-invariant $y^2$, together with $T$, obtains another $G$-invariant:
$$\left(\frac{1-q}{q}\cdot x+\frac{q^2-1}{q}\cdot y\right)^2$$
which produces the second generator $f_2=\left(x- (q+1) \cdot y\right)^2$ via dividing $\left(\frac{1}{q}-1\right)^2$.
\end{proof}

\begin{coro}\label{coro3.5}
Let $B(x,y)$ be the weight enumerator of a formally self-dual quantum code $Q$ of the type $((n,K))_q$. Then $B(x,y)$ is a polynomial expressed by $\{f_1,f_2\}$. 
\end{coro}

\begin{proof}
By the MacWilliams identity (\ref{Mac1}) and the fact that $B(x,y)=\frac{1}{K}\cdot B^\bot(x,y)$, we see that
$$
B(x,y)=B\left(\frac{1}{q}\cdot x+\frac{q^2-1}{q}\cdot y,\frac{1}{q}\cdot x-\frac{1}{q}\cdot y\right)=
B(\upsigma\cdot x,\upsigma\cdot y)=\upsigma\cdot B(x,y).
$$
This means that $B(x,y)$ is a $G$-invariant polynomial, thus, $B(x,y)\in \C[V]^G$. By Theorem \ref{thm3.3}, we see that
$B(x,y)\in \C[f_1,f_2]$.
\end{proof}

\section{Double Weight Enumerators of Self-dual Quantum Codes}\label{sec4}
\setcounter{equation}{0}
\renewcommand{\theequation}
{4.\arabic{equation}}
\setcounter{theorem}{0}
\renewcommand{\thetheorem}
{4.\arabic{theorem}}

\noindent This section is devoted  to describing the double  weight enumerator $C(x,y,z,w)$ of an arbitrary formally self-dual quantum code $Q$ of the type $((n,K))_q$. By the MacWilliams identity (\ref{Mac2}) and Proposition \ref{prop2.4}, we see that
\begin{equation}
\label{ }
C(x,y,z,w)=C\left(x+(q-1)y,x-y,\frac{z+(q-1)w}{q},\frac{z-w}{q}\right).
\end{equation}

We define $$\upsigma:=\begin{pmatrix}
\frac{1}{q}&\frac{q-1}{q}&0&0\\
\frac{1}{q}&-\frac{1}{q}& 0&0&\\
0&0&     1  &  q- 1 \\
 0&0&   1  & -1 \\
\end{pmatrix}.$$
Clearly, $\upsigma^2=\dia\left\{\frac{1}{q},\frac{1}{q},q,q\right\}$ and 
$$(\upsigma^{-1})^t:=\begin{pmatrix}
    1  &   1 & 0&0\\
    q-1  & -1&0&0 \\
    0&0&\frac{1}{q}&\frac{1}{q}\\
    0&0&\frac{q-1}{q}&-\frac{1}{q}
\end{pmatrix}.$$
Hence,
\begin{equation}
\label{eq4.2}
C(x,y,z,w)=C\left((\upsigma^{-1})^t(x),(\upsigma^{-1})^t(y),(\upsigma^{-1})^t(z),(\upsigma^{-1})^t(w)\right)=\upsigma\cdot C(x,y,z,w).
\end{equation}
Suppose that $G=\ideal{\upsigma}$ denotes the cyclic subgroup of $\GL_4(\C)$ generated by $\upsigma$ and $V$ denotes 
the standard $4$-dimensional representation of $G$ over $\C$. We write $\{x,y,z,w\}$ for the dual basis of $V^*$. Thus
$\C[V]=\C[x,y,z,w]$ and it follows from (\ref{eq4.2}) that 
$$C(x,y,z,w)\in\C[V]^G.$$
In other words, to describe the general shape of $C(x,y,z,w)$, we only need to compute the invariant ring $\C[V]^G.$

\begin{lem}\label{lem4.1}
Let $H=\ideal{\upsigma^2}$ be the subgroup of $G$ generated by $\upsigma^2$. Then $H$ is a normal subgroup of $G$ and 
$G/H\cong S_2.$
\end{lem}

\begin{proof}
Since $G$ is abelian and $H$ is a subgroup, $H$ is normal. To prove that $G/H\cong S_2$, we define a map $\upvarphi:G\ra S_2=\{0,1\}$ that maps $\upsigma^i$ to $0$ if $i$ is even, and maps $\upsigma^i$ to $1$ if $i$ is odd. Clearly, the map $\upvarphi$ is a surjective group homomorphism with the kernel $H$. Hence, $G/H=G/\ker(\upvarphi)\cong S_2$.
\end{proof}

Let us first compute $\C[V]^H$.

\begin{prop}\label{prop4.2}
$\C[V]^H=\C[x,y,z,w]^H=\C[xz,xw,yz,yw]$.
\end{prop}

\begin{proof}
Note that $\left(\upsigma^{-2}\right)^t=\dia\left\{q,q,\frac{1}{q},\frac{1}{q}\right\}$, thus the action of $\upsigma^2$ on $V^*$ is given by $$x\mapsto q\cdot x, y\mapsto q\cdot y, z\mapsto \frac{1}{q}\cdot z\textrm{ and }w\mapsto \frac{1}{q}\cdot w.$$
Clearly, $xz,xw,yz,yw$ are $H$-invariant. 

To prove that $\C[x,y,z,w]^H$ is generated by these four invariants, we first note that 
$\upsigma^2$ fixes any monomial $x^iy^jz^rw^t$, up to a nonzero scalar. This means that 
$\C[x,y,z,w]^H$ is generated by finitely many invariant monomials in $x,y,z,w$.
Consider an arbitrary monomial $f=x^iy^jz^rw^t\in \C[x,y,z,w]^H$ for $i,j,r,t\in\N$. Then
$$x^iy^jz^rw^t=f=\upsigma^2\cdot f=q^{i+j-r-t} \cdot x^iy^jz^rw^t$$
which implies that $q^{i+j-r-t}=1$ and thus 
\begin{equation}
\label{ }
i+j=r+t.
\end{equation}
This also means that the degree of $f$ must be even and we may assume that $\deg(f)\geqslant 2$. Apparently,  
$\{xz,xw,yz,yw\}$ spans the vector space of all invariants of degree $2$.

We use induction on the degree of $f$ to prove that $f$ is a polynomial in $xz,xw,yz,yw$. Suppose that $\deg(f)=2(i+j)\geqslant 4$.
Then $i+j\geqslant 2$. At least one element of $\{i,j\}$ is greater than or equal to $1$. The same statement holds for $\{r,t\}$. 
Without loss of generality, we may assume that $i,r\geqslant 1$. Then
$$f=x^iy^jz^rw^t=(xz)\cdot x^{i-1}y^jz^{r-1}w^t.$$
Note that $x^{i-1}y^jz^{r-1}w^t=\frac{f}{xz}$ is an $H$-invariant of degree $<\deg(f)$. 
Applying the induction hypothesis, we see that $x^{i-1}y^jz^{r-1}w^t\in \C[xz,xw,yz,yw]$. Hence,
$$f=x^iy^jz^rw^t=(xz)\cdot x^{i-1}y^jz^{r-1}w^t\in \C[xz,xw,yz,yw].$$
This proves that $\C[x,y,z,w]^H=\C[xz,xw,yz,yw]$.
\end{proof}

We define
\begin{equation}
\label{ }
U_1:=xz,~~U_2:=yw,~~V_1:=xw,~~V_2:=yz.
\end{equation}
Clearly, $U_1U_2-V_1V_2=0$. Moreover, this relation is the only relation among these generators. In other words,
 $\C[V]^H$ is a hypersurface of Krull dimension $3$. By Proposition \ref{prop4.2}, together with the fact $G/H$ is a finite group, we immediately derive 

\begin{coro}\label{coro4.3}
$\C[V]^G=\C[U_1,U_2,V_1,V_2]^{G/H}$ is of Krull dimension $3$.
\end{coro}

As $\C[U_1,U_2,V_1,V_2]$ is not polynomial, we use the technique appeared in \cite[Chapter 14, page 211]{CW11}
to compute $\C[V]^G$. More precisely, we introduce four variables $u_1,u_2,v_1,v_2$ of degree $1$ that correspond to
$U_1,U_2,V_1,V_2$ in $\C[x,y,w,z]$, respectively. 
Note that the action of $G/H$ on $\{U_1,U_2,V_1,V_2\}$ is stable. Thus this  induces an action of $G/H$ on  $\{u_1,u_2,v_1,v_2\}$ and so an action on the polynomial ring $\C[u_1,u_2,v_1,v_2]$. 
Since $G/H$ is linearly reductive, there exists a natural $G/H$-equivariant algebra surjection:
\begin{equation}
\label{eq4.5}
\uprho: \C[u_1,u_2,v_1,v_2]^{G/H}\ra \C[U_1,U_2,V_1,V_2]^{G/H}=\C[V]^G
\end{equation}
defined by sending $u_i\mapsto U_i$ and $v_i\mapsto V_i$.

We choose $\upsigma$ as the nontrivial left coset representative of $G$ over $H$. Recall that
 the action of $\upsigma$ on $V^*=\ideal{x,y,z,w}$ is given by
$$\upsigma: x\mapsto x+(q-1)\cdot y~, y\mapsto x-y~, z\mapsto  \frac{1}{q}\cdot z+\frac{q-1}{q}\cdot w~, w\mapsto  \frac{1}{q}\cdot z-\frac{1}{q}\cdot w.$$
This  induces an action of $G/H=\ideal{\upsigma H}$ on $\C[u_1,u_2,v_1,v_2]$ defined by
\begin{eqnarray*}
u_1 & \mapsto & \frac{1}{q} \cdot u_1+    \frac{(q-1)^2}{q} \cdot u_2+  \frac{q-1}{q}\cdot v_1+  \frac{q-1}{q}\cdot v_2 \\
u_2 & \mapsto & \frac{1}{q} \cdot u_1+    \frac{1}{q} \cdot u_2-  \frac{1}{q}\cdot v_1-  \frac{1}{q}\cdot v_2 \\
v_1 & \mapsto & \frac{1}{q} \cdot u_1+    \frac{1-q}{q} \cdot u_2-  \frac{1}{q}\cdot v_1+ \frac{q-1}{q}\cdot v_2 \\
v_2 & \mapsto & \frac{1}{q} \cdot u_1+    \frac{1-q}{q} \cdot u_2+ \frac{q-1}{q}\cdot v_1-  \frac{1}{q}\cdot v_2.
\end{eqnarray*}
We write $[\upsigma]$ for the resulting matrix of $\upsigma H$ on the vector space spanned by $\{u_1,u_2,v_1,v_2\}$. Then
$$[\upsigma]=\begin{pmatrix}
   \frac{1}{q}   & \frac{1}{q}& \frac{1}{q}& \frac{1}{q}\\
   \frac{(q-1)^2}{q}    &\frac{1}{q}  &\frac{1-q}{q}&\frac{1-q}{q}\\
   \frac{q-1}{q}   &-\frac{1}{q} &-\frac{1}{q} &\frac{q-1}{q}  \\
    \frac{q-1}{q}  &-\frac{1}{q}  &\frac{q-1}{q}&-\frac{1}{q}
\end{pmatrix}.$$
Note that $[\upsigma]$ is of order $2$ and $([\upsigma]^{-1})^t$ is similar with  $\uptau:=\dia\{1,1,-1,-1\}$. More precisely,
we define 
$$T:=\begin{pmatrix}
   \frac{q+1}{q}  &  \frac{(q-1)^2}{q} &\frac{q-1}{q}&\frac{q-1}{q}  \\
    1  & 1-q&q-1&q-1\\
   \frac{1-q}{q} &\frac{(q-1)^2}{q}&\frac{q-1}{q}&\frac{q-1}{q}\\
   1 &1-q&-q-1&q-1
\end{pmatrix}.$$
A direct verification shows that
\begin{equation}
\label{ }
[\upsigma]=T^{-1}\cdot \uptau\cdot T.
\end{equation}
Hence, the subgroup $N$ generated by $(\uptau^{-1})^t$ in $\GL_4(\C)$ and the subgroup generated by
 $\upsigma H$ in $G$ give rise to two equivalent representations of $G/H\cong S_2$, respectively.

Hence, by  Algorithm \ref{algm}, we may first compute $\C[u_1,u_2,v_1,v_2]^{N}$
and use the matrix $T$ to transfer the generating set of $\C[u_1,u_2,v_1,v_2]^{N}$ to a 
generating set of $\C[u_1,u_2,v_1,v_2]^{G/H}$. 

\begin{prop}\label{prop4.4}
The invariant ring $\C[u_1,u_2,v_1,v_2]^{N}=\C[u_1,u_2, v_1^2,v_2^2,v_1v_2]$ is hypersurface, subject to  the unique relation:
$$(v_1v_2)^2=v_1^2\cdot v_2^2.$$
\end{prop}

\begin{proof}
Note that $(\uptau^{-1})^t=\uptau=\dia\{1,1,-1,-1\}$, which fixes $u_i$ and maps $v_i$ to $-v_i$ for $i\in\{1,2\}$.
By Noether's bound theorem (see for example, \cite[Theorem 3.5.1]{CW11}), $\C[u_1,u_2,v_1,v_2]^{N}$ can be generated
by homogeneous polynomials of degree at most $|N|=2$. Thus, it suffices to consider an invariant monomial 
$f=v_1^iv_2^j$ with $i+j=2$. The three partitions of $2$: $(2,0), (0,2),$ and $(1,1)$ produce three invariant monomials:
$v_1^2,v_2^2,v_1v_2$, respectively. Therefore, $\C[u_1,u_2,v_1,v_2]^{N}$ can be generated by
$\{u_1,u_2, v_1^2,v_2^2,v_1v_2\}$.
\end{proof}

Combining Proposition \ref{prop4.4} and   Algorithm \ref{algm}, we obtain 

\begin{coro}\label{coro4.5}
$\C[u_1,u_2,v_1,v_2]^{G/H}$ is minimally generated by the following five polynomials:
\begin{eqnarray*}
f_1 &:= &  \frac{q+1}{q}\cdot u_1  +  \frac{(q-1)^2}{q}\cdot u_2 +\frac{q-1}{q}\cdot v_1+\frac{q-1}{q} \cdot v_2 \\
f_2 &: = & u_1+ (1-q)\cdot u_2+(q-1)\cdot v_1+(q-1)\cdot v_2\\ 
f_3 &:=& \left(\frac{1-q}{q}\cdot u_1  +  \frac{(q-1)^2}{q}\cdot u_2 +\frac{q-1}{q}\cdot v_1+\frac{q-1}{q} \cdot v_2\right)^2\\
f_4 &:=& (u_1+ (1-q)\cdot u_2-(q+1)\cdot v_1+(q-1)\cdot v_2)^2\\
f_5 &:=& \left(\frac{1-q}{q}\cdot u_1  +  \frac{(q-1)^2}{q}\cdot u_2 +\frac{q-1}{q}\cdot v_1+\frac{q-1}{q} \cdot v_2\right)\cdot\\
&&(u_1+ (1-q)\cdot u_2-(q+1)\cdot v_1+(q-1)\cdot v_2)
\end{eqnarray*}
subject to the unique relation: $f_5^2-f_3f_4=0$.
\end{coro}

Together Corollary \ref{coro4.3}, Corollary \ref{coro4.5}, and the surjective map $\uprho$ in (\ref{eq4.5}) gives 

\begin{coro}\label{coro4.6}
$\C[V]^G=\C[x,y,z,w]^{G}$ is  minimally generated by
\begin{eqnarray*}
g_1 &:= &  \frac{q+1}{q}\cdot xz  +  \frac{(q-1)^2}{q}\cdot yw +\frac{q-1}{q}\cdot xw+\frac{q-1}{q} \cdot yz \\
g_2 &: = & xz+ (1-q)\cdot yw+(q-1)\cdot xw+(q-1)\cdot yz\\ 
g_3 &:=& \left(\frac{1-q}{q}\cdot xz  +  \frac{(q-1)^2}{q}\cdot yw +\frac{q-1}{q}\cdot xw+\frac{q-1}{q} \cdot yz\right)^2\\
g_4 &:=& (xz+ (1-q)\cdot yw-(q+1)\cdot xw+(q-1)\cdot yz)^2\\
g_5 &:=& \left(\frac{1-q}{q}\cdot xz  +  \frac{(q-1)^2}{q}\cdot yw +\frac{q-1}{q}\cdot xw+\frac{q-1}{q} \cdot yz\right)\cdot\\
&&(xz+ (1-q)\cdot yw-(q+1)\cdot xw+(q-1)\cdot yz).
\end{eqnarray*}
\end{coro}

\begin{rem}{\rm
Note that $\C[V]^G$ is of Krull dimension $3$ while it is generated by $5$ invariant polynomials. Thus it is not hypersurface.
\hbo}\end{rem}

We close this section with the following description on the double weight enumerator of a formally self-dual quantum code.

\begin{coro}\label{coro4.8}
Let $C(x,y,z,w)$ be the double weight enumerator of a formally self-dual quantum code $Q$ of the type $((n,K))_q$. Then $C(x,y,z,w)$ is a polynomial in $g_1,g_2,\dots,g_5$. 
\end{coro}

\section{Complete Weight Enumerators of Self-dual Quantum Codes}\label{sec5}
\setcounter{equation}{0}
\renewcommand{\theequation}
{5.\arabic{equation}}
\setcounter{theorem}{0}
\renewcommand{\thetheorem}
{5.\arabic{theorem}}

\noindent Our experience in computational invariant theory shows that the complexity of an invariant ring $k[V]^G$ usually depends upon its Krull dimension, i.e., the dimension of $V$ as a $k$-vector space. 
Although computing high-dimensional invariant rings is more complicated than working with low-dimensional ones, studying the low-dimensional cases often provides significant insights into the structure of high-dimensional invariant rings; see for example, \cite{Che14, Che18, Che21, Ren24} and \cite{CR26}. 

This last section provides two explicit examples, exploring the complete weight enumerators of  formally self-dual quantum codes for $q=2$ and $3$, and demonstrating how difficult to calculate all complete weight enumerators of formally self-dual quantum codes.

Let us begin with the MacWilliams identity (\ref{Mac3}), for which we see that the complete weight enumerator $D(M)$
of a formally self-dual quantum code $Q$ of type $((n,K))_q$ satisfies the following equation:
\begin{equation}
\label{eq5.1}
D(M)=\frac{1}{K}\cdot D(M^\bot)
\end{equation}
where $M=(M_{\uplambda,\upmu})$ and $M^\bot=(M_{\uplambda',\upmu'}^\bot)$ denote
$q\times q$-matrices with entries 
$$M_{\uplambda',\upmu'}^\bot=\frac{1}{q}\cdot\sum_{\uplambda,\upmu\in\{0,1,\dots,q-1\}}
\upzeta_p^{\tr(\upalpha_{\uplambda'}\cdot \upalpha_{\upmu}-\upalpha_{\uplambda}\cdot \upalpha_{\upmu'})} M_{\uplambda,\upmu}.$$

\subsection{Example 1: $q=2$}

In this binary case, $p=q=2$ and $\upzeta_2=-1$. We assume that $\F_q=\{\upalpha_0=0,\upalpha_1=1\}$. Note that
$$M=\begin{pmatrix}
   M_{00}   & M_{01}   \\
   M_{10}   & M_{11} 
\end{pmatrix}.$$
To use $M_{\uplambda,\upmu}$ to express $M_{\uplambda,\upmu}^\bot$, we need the fundamental property of
the trace map:
\begin{equation}
\label{ }
\tr(a)=s\cdot a
\end{equation}
for all $a\in\F_p$ and where $q=p^s$; see \cite[Theorem 7.12 (iii)]{Wan12}. Hence,
\begin{eqnarray*}
M_{00}^\bot  & = & \frac{1}{2}\cdot\sum_{\uplambda,\upmu\in\{0,1\}}(-1)^{\tr(0)}\cdot M_{\uplambda,\upmu}= \frac{1}{2}\cdot\left(M_{00}+M_{01}+M_{10}+M_{11}\right) \\
M_{01}^\bot  & = & \frac{1}{2}\cdot\sum_{\uplambda,\upmu\in\{0,1\}}(-1)^{\tr(-\upalpha_\uplambda)}\cdot M_{\uplambda,\upmu}= \frac{1}{2}\cdot\left(M_{00}+M_{01}-M_{10}-M_{11}\right) \\
M_{10}^\bot  & = & \frac{1}{2}\cdot\sum_{\uplambda,\upmu\in\{0,1\}}(-1)^{\tr(\upalpha_\upmu)}\cdot M_{\uplambda,\upmu}= \frac{1}{2}\cdot\left(M_{00}-M_{01}+M_{10}-M_{11}\right) \\
M_{11}^\bot  & = & \frac{1}{2}\cdot\sum_{\uplambda,\upmu\in\{0,1\}}(-1)^{\tr(\upalpha_\upmu-\upalpha_\uplambda)}\cdot M_{\uplambda,\upmu}= \frac{1}{2}\cdot\left(M_{00}-M_{01}-M_{10}+M_{11}\right).
\end{eqnarray*}
We define
\begin{equation}
\label{ }
\upsigma:=\frac{1}{2}\cdot\begin{pmatrix}
    1  & 1  & 1&1\\
    1  & 1 &-1&-1\\
    1&-1&1&-1\\
    1&-1&-1&1
\end{pmatrix}.
\end{equation}
Then $\upsigma^2=I_4$ and $\upsigma^t=\upsigma$. Thus it follows from (\ref{eq5.1}) that
$$
D(M_{00},M_{01},M_{10},M_{11})=D(\upsigma(M_{00}),\upsigma(M_{01}),\upsigma(M_{10}),\upsigma(M_{11}))=\upsigma\cdot
D(M_{00},M_{01},M_{10},M_{11}).
$$
Let $G$ be the cyclic subgroup of $\GL_4(\C)$ generated by $\upsigma$ and $V$ denote the standard representation 
of $G$ over $\C$. Hence, the complete weight enumerator $D(M)$ can be viewed as a polynomial invariant in $\C[V]^G$.
To describe $D(M)$, we only need to find a homogeneous generating set for $\C[V]^G$.

We define $\uptau:=\dia\{1,1,1,-1\}$ and
$$T:=\begin{pmatrix}
    \frac{3}{2}  &  \frac{1}{2}  & \frac{1}{2} &\frac{1}{2} \\
    1  & -1 &3&-1\\
    1&-1&-1&3\\
    1&-1&-1&-1
\end{pmatrix}.$$
One may verify that
$$\upsigma=T^{-1}\cdot \uptau\cdot T.$$
It is easy to see that the invariant ring of $\uptau$ on $M_{00},M_{01},M_{10},M_{11}$ is equal to
\begin{equation}
\label{ }
\C[M_{00},M_{01},M_{10},M_{11}^2]
\end{equation}
which is a polynomial algebra over $\C$. By  Algorithm \ref{algm}, the following polynomials 
\begin{eqnarray*}
f_1 & := &\frac{3}{2}\cdot M_{00}+\frac{1}{2}\cdot M_{01}+\frac{1}{2}\cdot M_{10}+\frac{1}{2}\cdot M_{11} \\
f_2 & := & M_{00}-M_{01}+3\cdot M_{10}-M_{11}\\
f_3 & := & M_{00}-M_{01}-M_{10}+3\cdot M_{11}\\
f_4 & := & (M_{00}-M_{01}-M_{10}-M_{11})^2
\end{eqnarray*}
are $G$-invariant and generate $\C[V]^G$. In fact, we may replace $f_1$ by $\widetilde{f}_1:=2\cdot f_1$
 and obtain another generating set of $\C[V]^G$: $\{\widetilde{f}_1,f_2,f_3,f_4\}$.

This also completes the proof of the following result. 

\begin{thm}\label{thm5.1}
If $D(M)$ is the complete  weight enumerator of a formally self-dual quantum code $Q$ of the type $((n,K))_2$, then $D(M)$ is a polynomial in $\widetilde{f}_1,f_2,f_3,f_4$. 
\end{thm}

\begin{exam}\label{exam5.2}{\rm
Consider the Bell state code $Q_B$. By Example \ref{exam2.3}, it is a formally self-dual code  with the complete weight enumerator $D(M)=\frac{1}{2}\left(M_{11}^2 + M_{10}^2 + M_{01}^2 + M_{00}^2\right)$. A direct computation shows that
$$(\widetilde{f}_1)^2 - \widetilde{f}_1f_2 - \widetilde{f}_1f_3 + f_2^2 + f_2f_3 + f_3^2 + 2f_4 - 16\cdot D(M)=0$$
which implies that $D(M)$ can be algebraically  expressed by  $\widetilde{f}_1,f_2,f_3,f_4$. 
\hbo}\end{exam}

\subsection{Example 2: $q=3$} In this case, $p=q=3$. Throughout this subsection, we 
write $\upomega$ for $$\upzeta_3=e^{\frac{2\pi\sqrt{-1}}{3}},$$ for the sake of simplicity. Thus, $\upomega^3=1$ and
$$\upomega^2+\upomega+1=0.$$

We assume that 
$\F_q=\{\upalpha_0=0,\upalpha_1=1,\upalpha_2=2\}$. By (\ref{eq5.1}), we see that
\begin{eqnarray*}
M_{00}^\bot  & = & \frac{1}{3}\cdot\sum_{\uplambda,\upmu\in\{0,1,2\}}\upomega^{\tr(0)}\cdot M_{\uplambda,\upmu}\\
&=& \frac{1}{3}\cdot\left(M_{00}+M_{01}+M_{02}+M_{10}+M_{11}+M_{12}+M_{20}+M_{21}+M_{22}\right) \\
M_{01}^\bot  & = & \frac{1}{3}\cdot\sum_{\uplambda,\upmu\in\{0,1,2\}}\upomega^{\tr(2\cdot \upalpha_{\uplambda})}\cdot M_{\uplambda,\upmu}\\
&=& \frac{1}{3}\cdot\left(M_{00}+M_{01}+M_{02}+\upomega^2\cdot M_{10}+\upomega^2\cdot M_{11}+\upomega^2\cdot M_{12}+\upomega\cdot M_{20}+\upomega\cdot M_{21}+\upomega\cdot M_{22}\right) \\
M_{02}^\bot  & = & \frac{1}{3}\cdot\sum_{\uplambda,\upmu\in\{0,1,2\}}\upomega^{\tr(\upalpha_{\uplambda})}\cdot M_{\uplambda,\upmu}\\
&=& \frac{1}{3}\cdot\left(M_{00}+M_{01}+M_{02}+\upomega\cdot M_{10}+\upomega\cdot M_{11}+\upomega\cdot M_{12}+\upomega^2\cdot M_{20}+\upomega^2\cdot M_{21}+\upomega^2\cdot M_{22}\right) \\
M_{10}^\bot  & = & \frac{1}{3}\cdot\sum_{\uplambda,\upmu\in\{0,1,2\}}\upomega^{\tr(\upalpha_{\upmu})}\cdot M_{\uplambda,\upmu}\\
&=& \frac{1}{3}\cdot\left(M_{00}+\upomega\cdot M_{01}+\upomega^2\cdot M_{02}+M_{10}+
\upomega\cdot M_{11}+\upomega^2\cdot M_{12}+M_{20}+\upomega\cdot M_{21}+\upomega^2\cdot M_{22}\right) \\
M_{11}^\bot  & = & \frac{1}{3}\cdot\sum_{\uplambda,\upmu\in\{0,1,2\}}\upomega^{\tr(\upalpha_{\upmu}-a_{\uplambda})}\cdot M_{\uplambda,\upmu}\\
&=& \frac{1}{3}\cdot\left(M_{00}+\upomega\cdot M_{01}+\upomega^2\cdot M_{02}+\upomega^2\cdot M_{10}+
M_{11}+\upomega\cdot M_{12}+\upomega\cdot M_{20}+\upomega^2\cdot M_{21}+M_{22}\right) \\
M_{12}^\bot  & = & \frac{1}{3}\cdot\sum_{\uplambda,\upmu\in\{0,1,2\}}\upomega^{\tr(\upalpha_{\upmu}+\upalpha_{\uplambda})}\cdot M_{\uplambda,\upmu}\\
&=& \frac{1}{3}\cdot\left(M_{00}+\upomega\cdot M_{01}+\upomega^2\cdot M_{02}+\upomega\cdot M_{10}+
\upomega^2\cdot M_{11}+M_{12}+\upomega^2\cdot M_{20}+M_{21}+\upomega\cdot M_{22}\right) \\
M_{20}^\bot  & = & \frac{1}{3}\cdot\sum_{\uplambda,\upmu\in\{0,1,2\}}\upomega^{\tr(2\cdot \upalpha_{\upmu})}\cdot M_{\uplambda,\upmu}\\
&=& \frac{1}{3}\cdot\left(M_{00}+\upomega^2\cdot M_{01}+\upomega\cdot M_{02}+M_{10}+
\upomega^2\cdot M_{11}+\upomega\cdot M_{12}+M_{20}+\upomega^2\cdot M_{21}+\upomega\cdot M_{22}\right) \\
M_{21}^\bot  & = & \frac{1}{3}\cdot\sum_{\uplambda,\upmu\in\{0,1,2\}}\upomega^{\tr(2\cdot \upalpha_{\upmu}-\upalpha_{\uplambda})}\cdot M_{\uplambda,\upmu}\\
&=& \frac{1}{3}\cdot\left(M_{00}+\upomega^2\cdot M_{01}+\upomega\cdot M_{02}+\upomega^2\cdot M_{10}+
\upomega\cdot M_{11}+M_{12}+\upomega\cdot M_{20}+M_{21}+\upomega^2\cdot M_{22}\right) \\
M_{22}^\bot  & = & \frac{1}{3}\cdot\sum_{\uplambda,\upmu\in\{0,1,2\}}\upomega^{\tr(\upalpha_{\uplambda}-\upalpha_{\upmu})}\cdot M_{\uplambda,\upmu}\\
&=& \frac{1}{3}\cdot\left(M_{00}+\upomega^2\cdot M_{01}+\upomega\cdot M_{02}+\upomega\cdot M_{10}+
M_{11}+\upomega^2\cdot M_{12}+\upomega^2\cdot M_{20}+\upomega\cdot M_{21}+M_{22}\right).
\end{eqnarray*}

Consider the following $9\times 9$ matrix:
$$\upsigma:=\frac{1}{3}\cdot
\begin{pmatrix}
 1& 1&1&1&1&1&1&1& 1  \\
  1&1 &1&\upomega&\upomega&\upomega&\upomega^2&\upomega^2& \upomega^2  \\
  1& 1&1&\upomega^2&\upomega^2&\upomega^2&\upomega&\upomega& \upomega  \\
 1&\upomega^2 &\upomega&1&\upomega^2&\upomega&1&\upomega^2&\upomega   \\
1 &\upomega^2 &\upomega&\upomega&1&\upomega^2&\upomega^2&\upomega&1   \\ 
  1 &\upomega^2 &\upomega&\upomega^2&\upomega&1&\upomega&1&\upomega^2   \\
  1 &\upomega &\upomega^2&1&\upomega&\upomega^2&1&\upomega& \upomega^2  \\
  1 &\upomega &\upomega^2&\upomega&\upomega^2&1&\upomega^2&1& \upomega  \\
1 &\upomega &\upomega^2&\upomega^2&1&\upomega&\upomega&\upomega^2& 1  \\ 
\end{pmatrix}$$
and $\updelta:=(\upsigma^{-1})^t$. Note that $\updelta^2=\upsigma^2=I_9$. Thus the standard 
$9$-dimensional representation $V$ of the group $G:=\ideal{\updelta}$ is a faithful representation of $S_2$ over $\C$.
It follows from (\ref{eq5.1}) that
\begin{equation}
\label{ }
D(M_{\uplambda,\upmu})\in\C[V]^G=\C[M_{\uplambda,\upmu}\mid \uplambda,\upmu\in\{0,1,2\}]^G.
\end{equation}

Note that $\upomega^2+\upomega+1=0$ and a direct computation shows that the matrix 
$$T:=\begin{pmatrix}
 1&-\upomega-1&\upomega&1&-\upomega-1&\upomega&4&-\upomega-1& \upomega   \\
  1&\upomega&-\upomega-1&-\upomega-1&4&\upomega&\upomega&-\upomega-1&1    \\
  1&1&4&\upomega&\upomega&\upomega&-\upomega-1&-\upomega-1&-\upomega-1    \\
   \frac{4}{3} &\frac{1}{3}&\frac{1}{3}&\frac{1}{3}&\frac{1}{3}&\frac{1}{3}&\frac{1}{3}&\frac{1}{3}&\frac{1}{3}    \\
   1&-\upomega-1&\upomega&-\upomega-1&\upomega&1&\upomega&4&-\upomega-1    \\
   1&-\upomega-1&\upomega&\upomega&1&-\upomega-1&-\upomega-1&\upomega&4    \\  
   1&1&-2&\upomega&\upomega&\upomega&-\upomega-1&-\upomega-1&-\upomega-1    \\
    -\frac{2}{3}&\frac{1}{3}&\frac{1}{3}&\frac{1}{3}&\frac{1}{3}&\frac{1}{3}&\frac{1}{3}&\frac{1}{3}&\frac{1}{3}    \\
    1&\upomega&-\upomega-1&-\upomega-1&-2&\upomega&\upomega&-\upomega-1&1    \\
\end{pmatrix}$$
makes
$$\updelta=T^{-1}\cdot \uptau\cdot T$$
holds, where $\uptau:=\dia\{1,1,1,1,1,1,-1,-1,-1\}$. It is not difficult to see that
$$\C[M_{\uplambda,\upmu}\mid \uplambda,\upmu\in\{0,1,2\}]^{\ideal{\uptau}}$$
can be  minimally generated by
$$\A:=\left\{M_{\uplambda,\upmu}, M_{20}^2, M_{21}^2,M_{22}^2, M_{20}M_{21}, M_{20}M_{22}, M_{21}M_{22}\mid \uplambda\in\{0,1\}, \upmu\in\{0,1,2\}\right\}$$
which is isomorphic to the following tensor product 
$$\C[M_{\uplambda,\upmu}\mid \uplambda\in\{0,1\}, \upmu\in\{0,1,2\}]\otimes_\C \C[M_{20}^2, M_{21}^2,M_{22}^2, M_{20}M_{21}, M_{20}M_{22}, M_{21}M_{22}].$$
Note that latter component of this tensor product is isomorphic to the second Veronese subring of 
the polynomial algebra in three variables.

Define a vector $W :=(M_{00}, M_{01}, M_{02}, M_{10}, M_{11}, M_{12},M_{20}, M_{21}, M_{22})$ and for $i\in \{1,2,\dots,9\}$, we use $f_i$ to denote the dot product of the $i$-th row of  $T$ and $W^t$. By  Algorithm \ref{algm}, we may use the matrix $T$ defined above to transfer the $12$ elements in $\A$ to a homogeneous generating set  of $\C[V]^G$:
\begin{equation}
\label{ }
\B:=\{f_i, g_j,h_r\mid 1\leqslant i\leqslant 6, 1\leqslant j,r\leqslant 3\},
\end{equation} 
where $g_j:=(f_{6+j})^2$, $h_1:=f_7f_8, h_2:=f_7f_9$, and $h_3:=f_8f_9$. Therefore, 

\begin{thm}\label{thm5.2}
If $D(M)$ is the complete  weight enumerator of a formally self-dual quantum code $Q$ of the type $((n,K))_3$, then $D(M)$ is a polynomial in elements of $\B$. 
\end{thm}

\vspace{2mm}
\noindent \textbf{Acknowledgements}. 
This research was partially supported by NNSF of China (Grant No. 12561003).
%The second author would like to thank Professor Runxuan Zhang for her encouragement and help. 

%The author would like to thank the anonymous referees and the editor for their careful reading, constructive comments, and suggestions. 

%%%%%%%%%%%%%%%%%%%%%%%%%%References%%%%%%%%%%%%%%%%%%%%%%%%

\raggedright
\end{document}